\documentclass[a4paper,USenglish,cleveref,autoref, thm-restate]{lipics-v2021}
\usepackage{tikz}

\pdfoutput=1 
\hideLIPIcs  

\newcommand{\eps}{\varepsilon}
\newcommand{\var}{\hbox{\rm var}\,}

\newcounter{assumption-counter}
\newtheorem{assumption}[assumption-counter]{\sc Assumption}

\nolinenumbers 

\title{The Geometry of Cyclical Social Trends
} 

\author{Bernard {Chazelle}}{Department of Computer Science, Princeton University, United States}{chazelle@cs.princeton.edu}{https://orcid.org/0000-0001-8542-0247}{}

\author{Kritkorn {Karntikoon}}{Department of Computer Science, Princeton University, United States}{kritkorn@princeton.edu}{https://orcid.org/0000-0002-6398-3097}{}

\author{Jakob {Nogler}}{Department of Computer Science, ETH Zurich, Switzerland}{jnogler@ethz.ch}{https://orcid.org/0009-0002-7028-2595}{}

\authorrunning{B. Chazelle, K. Karntikoon, and J. Nogler} 

\Copyright{Bernard Chazelle, Kritkorn Karntikoon, and Jakob Nogler} 



\funding{This work was supported in part by NSF grant CCF-2006125.}

\EventEditors{John Q. Open and Joan R. Access}
\EventNoEds{2}
\EventLongTitle{42nd Conference on Very Important Topics (CVIT 2016)}
\EventShortTitle{CVIT 2016}
\EventAcronym{CVIT}
\EventYear{2016}
\EventDate{December 24--27, 2016}
\EventLocation{Little Whinging, United Kingdom}
\EventLogo{}
\SeriesVolume{42}
\ArticleNo{23}

\begin{document}

\maketitle

\begin{abstract}
We investigate the emergence of periodic behavior in opinion dynamics
and its underlying geometry. 
For this, we use a bounded-confidence model with contrarian agents
in a convolution social network. 
This means that agents adapt their opinions by interacting with
their neighbors in a time-varying social network. Being contrarian,
the agents are kept from reaching consensus. This is the key feature 
that allows the emergence of cyclical trends.
We show that the systems either converge to nonconsensual
equilibrium or are attracted to
periodic or quasi-periodic orbits. 
We bound the dimension of the attractors
and the period of cyclical trends.
We exhibit instances where each orbit is
dense and uniformly distributed within its attractor.
We also investigate the case of randomly changing
social networks. 
\end{abstract}

\section{Introduction}\label{intro}

Much of the work in the area of opinion dynamics has focused
on consensus and polarization~\cite{bernardoC, FagnaniF}.
Typical questions include:
How do agents come to agree or disagree?
How do exogenous forces drive them to consensus? How long does it take
for opinion formation to settle?  Largely left out of the discussion has been
the emergence of \emph{cyclical trends}. Periodic 
patterns in opinions and preferences is a complex,
multifactorial social phenomenon beyond the scope of this work~\cite{Lu5}.
A question worth examining, however, is whether
the process conceals deeper mathematical structure.
The purpose of this work is to show that it is, indeed, the case.

This work began with a thought experiment and a computer simulation.
The latter revealed highly unexpected behavior, which in turn compelled 
us to search for an explanation.
Our main result is a proof that adding a simple contrarian rule
to the classic bounded-confidence model suffices to produce quasi-periodic trajectories.
The model is a slight variant of the classic \emph{HK} framework:
a finite collection of agents hold opinions on several topics, which they update 
at discrete time steps by consulting
their neighbors in a (time-varying) social network. 
The modification is the addition of a simple repulsive force field that
keep agents away from tight consensus. The idea is partly inspired by
swarming dynamics. For example, birds refrain from 
flocking too closely. Likewise, near-consensus on a large enough scale tends to induce
contrarian reactions among agents~\cite{alford, heese22}.
Some political scientists have pointed to contrarianism as one of the reasons
for the closeness of some national elections~\cite{galam19, galam20}.

One of the paradoxical observations we sought to elucidate was
why cyclic trends in social networks seem oblivious to the initial
opinions of one's friends: specifically, it is not specific distributions
of initial opinions that produce oscillations but, rather,
the recurrence of certain symmetries in the networks. We prove that
the condition is sufficient (though its necessity is still open).
Another mystery
was why contrarian opinions tend to orbit toward an attractor whose
dimensionality is \emph{independent}  
of the number of opinions held by a single agent. 
These attracting sets are typically Minkowski sums of ellipses.
They emerge algorithmically and constitute a natural focus
of interest in distributed computational geometry. 

Our inquiry builds on the pioneering work
of French~\cite{french56},
DeGroot~\cite{degroot1974},
Friedkin \& Johnsen~\cite{fj1990},
and Deffuant et al.~\cite{deffuant}.
The model we use is a minor modification of the
\emph{bounded-confidence model} model~\cite{blondelHT09,hegselmanK}.
A \emph{Hegselmann-Krause (HK)} system consists of 
$n$ agents, each one represented by a point in ${\mathbb R}^d$.
The $d$ coordinates for each agent $i$ represent their current opinions
on $d$ different topics: thus, $d$ is the dimension of the opinion space.
At any (discrete) time, each agent~$i$ moves to the mass center of
the agents within a fixed distance $r_i$, 
which represents its radius of influence (Fig.~\ref{figHK}).
This step is repeated ad infinitum.
Formally, the agents are positioned at $x_1(t),\ldots, x_n(t)\in {\mathbb R}^d$ at time
$t$ and for any $t=0,1,2,\ldots \,$,
\begin{equation}\label{HK1}
x_i(t+1)= \frac{1}{| \mathcal{N}_i(t)|} \sum_{j\in \mathcal{N}_i(t)} x_j(t)\, ,
\;
\text{with } \mathcal{N}_i(t)=\Bigl\{\, 1\leq j\leq n\,:\, \big\|x_i(t)-x_j(t)\big\|_2\leq r_i\, \Bigr\}.
\end{equation}

\begin{figure}[htb]
\centering
\includegraphics[width=11cm]{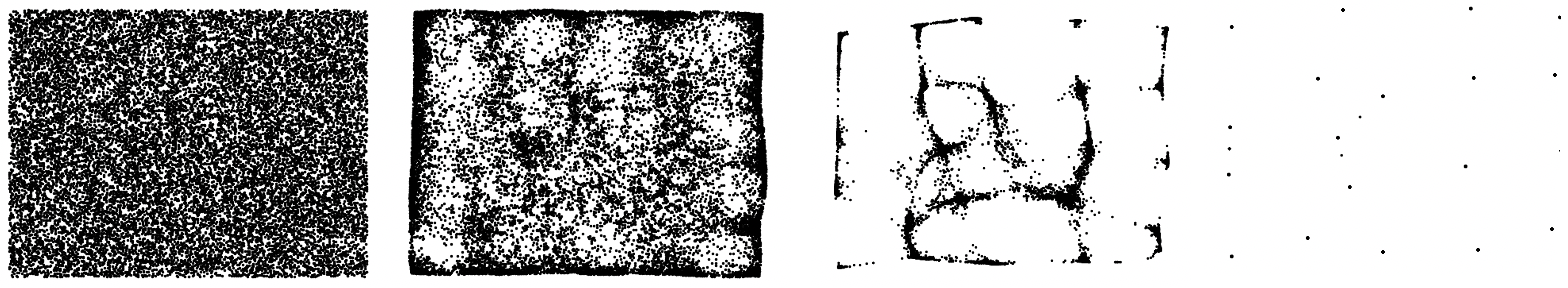}
\caption{The evolution of 20,000 random points in an \emph{HK} system.}
\label{figHK}
\end{figure}

Interpreting each $\mathcal{N}_i(t)$ as the set of neighbors of agent $i$
defines the \emph{social network} $G_t$ at time~$t$.
In the special case where all the radii of influence are equal $(r_i=R)$, convergence into
fixed-point clusters occurs within a polynomial number of steps~\cite{BBCN13, EBNT15, martinssonHK}.
Computer simulation suggests that the same remains true even when the radii differ
but a proof has remained elusive.
For cyclical trends to emerge, the social networks require a higher degree of underlying structure.
In this work, we assume vertex transitivity (via Cayley graphs), which stipulates that agents cannot
be distinguished by their local environment. Before defining the model formally in the next section,
we summarize our main findings.
\begin{itemize}
\item
Undirected networks always drive the agents
to nonconsensual convergence, ie, to fixed points at which they ``agree to disagree.''
For their behavior to become periodic or quasi-periodic,
the social networks need to be directed.
We prove that such systems either converge or are attracted to
periodic or quasi-periodic orbits. We give precise formulas for the orbits.

\item
We investigate the geometry of the attractors (Fig.~\ref{fig-attractors}).
We bound the \emph{rotation number}, which
indicates the speed at which (quasi)-periodic opinions undergo a full cycle.
We exhibit instances where each limiting orbit forms a set that is
dense and, in fact, uniformly distributed on its attractor.

\item
We explore the case of social networks changing randomly at each step.
We prove the surprising result that the dimension
of the attractor can \emph{decrease} because of the randomization.
This is a rare case where adding entropy to a system
can reduce its dimensionality. 
\end{itemize}
The dynamics of contrarian views has been studied 
before~\cite{alford, eekhoff, ferraz21, galam19, galam20, heese22, muslimKN}
but, to our knowledge, not for the purpose of explaining cyclical trends.
Our mathematical findings can be viewed as a grand generalization of the affine-invariant
evolution of planar polygons studied 
in~\cite{bruckstein95, davis94, elmachVL10, kostadinov17}.

\begin{figure}[htb]
\centering
\includegraphics[width=11cm]{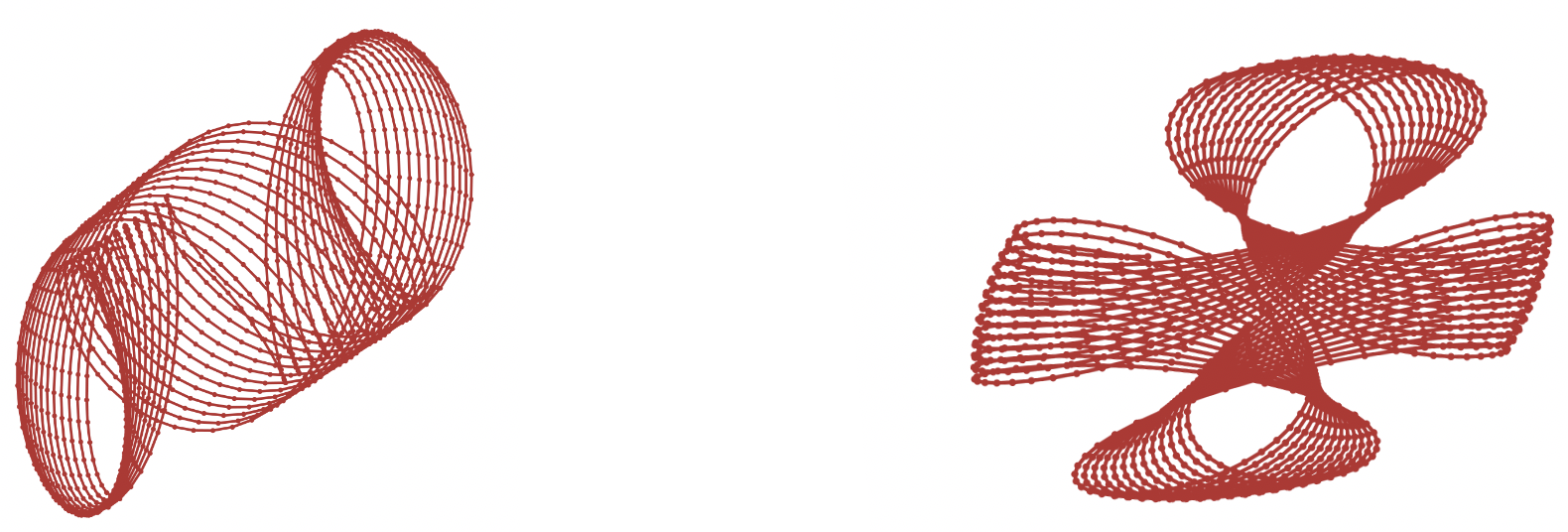}
\caption{Typical attractors.}
\label{fig-attractors}
\end{figure}

\section{Contrarian Opinion Dynamics}\label{secCOD}

The social network is a time-dependent Cayley graph over an abelian group.
All finite abelian groups are isomorphic to a direct sum of
cyclic groups $(\mathbb{Z}/n_1\mathbb{Z}) \oplus \cdots \oplus (\mathbb{Z}/n_m \mathbb{Z})$. 
For notational convenience, we set $n_i=n$.
We regard the toral grid $V=  (\mathbb{Z}/n\, \mathbb{Z})^m$ as a vector space,
and we write $N= |V|= n^m$.
Let $x_v(t)$ be the position of agent $v$ in $\mathbb{R}^d$ at time $t$.
We fix $x_v(0)$ and abbreviate it as $x_v$.
Choose $p$ such that $1/N <p< 1$
and let $(C_t)_{t\geq 0}$ be an infinite sequence of subsets of $V$.
For technical convenience, we assume that each set $C_t$ spans
the vector space $V$; hence $|C_t|\geq m$.
In the spirit of \emph{HK} systems,
we define the dynamics as follows:
for $t=0,1,\ldots,$
\begin{equation}\label{dyn}
x_v(t+1)= p x_v(t) + \frac{1-p}{|C_t|} \sum_{w\in v+C_t} x_w(t).
\end{equation}
Because of the presence of the ``self-confidence'' weight $p$, we may
assume that the \emph{convolution set} $C_t$ does not contain the origin $\mathbf{0}$.
If we view each $x_v(t)$ as a row vector in $\mathbb{R}^d$,
the update~(\ref{dyn}) specifies an $N$-by-$N$ stochastic matrix $F_{C_t}$.
Let $x(t)$ denote the $N$-by-$d$ matrix whose rows
are the $N$ agent positions $x_v(t)$, for $v\in V$.
We have $x(t+1) = F_{C_t}x(t)$.
The matrix $F_{C_t}$ may not be symmetric but it is always doubly-stochastic.
This means that the mass center $\mathbf{1}^\top x(t)/N$
is time-invariant.  Since the dynamics itself
is translation-invariant, we are free to move the mass center to the origin,
which we do by assuming $\mathbf{1}^\top x= \mathbf{0}^\top$, where $x$ denotes $x(0)$.

Obviously, some initial conditions are uninteresting: for example, $x=\mathbf{0}$. 
For this reason, we choose $x$ randomly; specifically,
each $x_v$ is picked iid from the $d-$dimensional
normal distribution $\mathcal{N}(\mathbf{0},1)$.
In the following, we use the phrase ``with high probability,'' 
to refer to an event occurring with probability at least $1-\eps$, for
any fixed $\eps>0$.  Once we've picked the matrix $x$ randomly,
we place the mass center of the agents at the origin by subtracting
its displacement from the origin:
$x\leftarrow x - \frac{1}{N}\mathbf{1}\mathbf{1}^\top x$.

The agents will be attracted
to the origin to form a single-point cluster of consensus in the limit.
Responding to their contrarian nature, the agents will restore
mutual differences by boosting the own opinions.  For that reason
we consider the scaled dynamics:
$y(0)=x$ and, for $t\geq 0$,
\begin{equation}\label{matrix-dyn}
y(t+1) = \xi_t F_{C_t}y(t),
\end{equation} 
where $\xi_t$ is chosen so that the diameter of the system remains
roughly constant. Since scaling
leaves the salient topological and geometric properties
of the dynamics unchanged, the precise definition of $\xi_t$
can vary to fit analytical (or even visual) needs.

\subsection{Preliminaries}

We define the directed graph $G_{C_t}=(V, E_t)$ at time $t\geq 0$,
where $E_t= \bigcup\, \{(v,v+c)\,|\, v\in V, c\in C_t\}$ and $N_v= \{(v,w)\,|\, w\in v+ C_t\subseteq V\}$.
For clarity, we drop the subscript $t$ for the remainder of this section;
so we write $C$ for $C_t$.

\begin{lemma}\label{span-SC}
The convolution set $C$ spans the vector space $V$ if and only if
the graph $G_C$ is strongly connected.
\end{lemma}

\begin{proof}
If $C$ spans $V$, then for any pair $u,v\in V$, there exist $a_h\in \mathbb{Z}/n\mathbb{Z}$,
for each $h\in C$, such that $v-u= \sum_{h\in C} a_h h$. The right-hand side specifies
$\sum_h a_h$ edges (sum taken over $\mathbb{N}$) that form a path from $u$ to $v$ ;
therefore $G$ is strongly connected.
Conversely, assuming the latter, there is a path from $u$ to $v$: 
$(w_1,w_2), \ldots, (w_{k-1},w_k)$, with $w_1=u$ and $w_k=v$.
Thus, $v-u= \sum_i  c_i$, where $c_i= w_i - w_{i-1}\in C$; therefore
$C$ spans $V$.
\end{proof}

Our assumption about $C$ implies that each $G_{C}$ is strongly connected.
The presence of the weight $p>0$ in~(\ref{dyn})
ensures that the diagonal of $F_{C}$ is
positive. Together with the strong connnectivity assumption, 
this makes the matrix $F_{C}$ primitive, meaning that $F_{C}^k>0$, for some $k>0$.
By the Perron-Frobenius theorem~\cite{seneta06},
all the eigenvalues of $F_{C}$ lie strictly inside the unit circle in $\mathbb{C}$,
except for the dominant eigenvalue~1, which has multiplicity~$1$.
For any $u,v\in V$, we write
$\psi_u^v = \omega^{\langle u,v\rangle}$, where $\omega:= e^{2\pi i/n}$.
We define the vector $\psi^v= (\psi_u^v\,|\, u\in V)$
and easily verify that $\{\psi^v\,|\, v\in V\}$ forms an 
orthogonal eigenbasis for~$F_{C}$.
The eigenvalue $\lambda_v$ corresponding to $\psi^v$ satisfies
\[
\lambda_v \psi_u^v
= p \psi_u^v +  \frac{1-p}{|C|} \sum_{w\in  u+C} \psi_w^v
= p \psi_u^v +  \frac{1-p}{|C|} \Big( \sum_{h\in C} \psi_h^v \Big) \psi_u^v \, .
\]
We conclude:

\begin{lemma}\label{eigenvalues}
Each $v\in V$ corresponds to a distinct eigenvector
$\psi^v$, which together form an orthogonal basis for $\mathbb{C}^N$.
The corresponding eigenvalue is given by
\[
\lambda_v= p +  \frac{1-p}{|C|}  \sum_{h\in C} \omega^{\langle v, h\rangle}.
\]
\end{lemma}

We define $\lambda= \max_{v\in V}\{ |\lambda_v|<1\}$ and 
denote by $W=\{v\in V : |\lambda_v| = \lambda \}$ the set of 
subdominant eigenvectors.
The argument of $\lambda_v$ plays a key role in our discussion,
so we define~$\theta_v$ such that $\lambda_v = |\lambda_v|\omega^{\theta_v}$,
with $\theta_v\in (-n/2,n/2]$.
By~(\ref{second-eigen}), $\lambda_v\neq 0$ for $v\in W$, so $\theta_v$ is well defined.

\subsection{The evolution of opinions}\label{sec-evol}

We begin with the case of a fixed convolution set $C_t=C$.
The initial position of the agents is expressed in eigenspace as
$x= \frac{1}{N} \sum_{v\in V} \psi^v (\psi^v)^\mathrm{H} x$.
Let $z_v$ denote the row vector 
$(\psi^v)^\mathrm{H} x= \sum_{u\in V} \omega^{-\langle v,u\rangle} x_u$.
Because $(\psi^v)^\mathrm{H} x = \mathbf{1}^\top x= \mathbf{0}^\top$, for $v= \mathbf{0}\in V$,
\begin{equation}\label{x(t)=}
x(t)= \frac{1}{N} \sum_{v\in V\setminus \{\mathbf{0}\}} 
        \lambda_v^t   \psi^v z_v.
\end{equation}

\begin{lemma}\label{eigen-whp}
With high probability, for all $v\neq \mathbf{0}$,
\[
\Omega\big( \sqrt{1/N} \, \big) = \|z_v\|_2 = O\big(\sqrt{dN\log dN}\, \big).
\]
\end{lemma}
\begin{proof}
Let $a= (a_u)_{u\in V}$ be the first column of the matrix $x$.
For each $u\in V$, by the initialization of the system,
$a_u = \zeta_u - \delta$, where $\zeta_u\sim \mathcal{N}(0,1)$ 
and $\delta= \frac{1}{N}\mathbf{1}^\top \zeta$.
Given $v\neq \mathbf{0}$, $\psi^v$ is orthogonal to $\psi^\mathbf{0}= \mathbf{1}$;
hence
$
(\psi^v)^\mathrm{H}a = (\psi^v)^\mathrm{H} (\zeta - \delta \mathbf{1}) = (\psi^v)^\mathrm{H} \zeta
$.
Since the random vector $\zeta$ is unbiased and $|\omega^{-\langle v,u\rangle}|=1$,
it follows that
$\var \big[ (\psi^v)^\mathrm{H}a \big] =  \sum_{u\in V} \var \zeta_u = N
$.
Thus, the first coordinate $z_{v,1}$ of $z_v$ is of the form
$a+ib$, where $a$ and $b$ are sampled (not independently)
from $\mathcal{N}(0, \sigma_1^2)$ and $\mathcal{N}(0, \sigma_2^2)$,
respectively, such that $\sigma_1^2+ \sigma_2^2=N$.
Thus,  $|z_{v,1}|\leq \delta$ with probability at most $2\delta /\sqrt{\pi N}$.
Conversely, by the inequality erfc$(z)\leq e^{-z^2}$ for $z>0$, we find that 
$|z_{v,1}|=O(\sqrt{N\log (dN/\eps)}\,)$, with probability at least $1-\eps/dN$,
for any $0<\eps<1$; hence
$\|z_v\|_2= O(\sqrt{d N\log (dN/\eps)}\,)$, with probability at least $1-\eps/N$.
Setting $\delta= \eps \sqrt{\pi /4N}$ and using a union bound completes the proof.
\end{proof}

We upscale the system by setting $\xi_t= 1/\lambda$; hence $y(t+1)= y(t)/\lambda$.

\begin{theorem}\label{y(t)=thm}
Let $a_h$ and $b_h$ be the row vectors whose $u$-th coordinates ($u\in V$)
are $\cos (2\pi \langle h, u \rangle /n)$
and $\sin (2\pi \langle h, u\rangle /n)$, respectively.
With high probability, for each $v\in V$, the agent $v$ is attracted to the 
trajectory of $y_v^*(t)$, where
\begin{equation}\label{y-abx}
y_v^*(t) = \frac{1}{N}
\sum_{h\in W}
\left( \cos \frac{2\pi  (t \theta_h  + \langle h, v \rangle)}{n}  \, , \, 
           \sin  \frac{2\pi (t \theta_h +  \langle h, v \rangle)}{n} \right)
     \begin{pmatrix} 
              a_h \\ b_h 
     \end{pmatrix} x .
\end{equation}     
Let $\mu:= \max\{ |\lambda_v|/\lambda <1\}$ be the third largest (upscaled) eigenvalue, measured
in distinct moduli. The error of the approximation decays exponentially fast
as a function of $\mu$:
\[
\frac{ \|y_v^*(t)- y_v(t)\|_F }{ \|y_v(t) \|_F }
  =  O\big( \mu^t N^2 \sqrt{d\log dN}\,\big).
\]
\end{theorem}
\begin{proof}
Since the eigenvalues sum up to tr$F_C= pN$ and 1 has
multiplicity 1, we have $pN\leq 1+ (N-1)\lambda$; hence, by $p>1/N$,
\begin{equation}\label{second-eigen}
\lambda\geq \frac{pN-1}{N-1}>0.
\end{equation}
Writing $\mu_v= \lambda_v / \lambda$ and $\mu= \max\{|\mu_v|<1\}$,
we have $|\mu_v|=1$ for $v\in W$; recall that
$W=\{v\in V : |\lambda_v| = \lambda \}$.
By~(\ref{x(t)=}), it follows that
\begin{equation}\label{y(t)=zv-eta}
y(t) =  \frac{1}{N} \sum_{v\in W}
                   \mu_v^t  \psi^v  z_v +  \eta(t),
\end{equation}
where,  by Lemma~\ref{eigen-whp}, with high probability,
\begin{align*}
    {\| \eta(t) \|}_F
&= \Big\|   \frac{1}{N} 
     \sum_{v\in V\setminus (W\cup \{\mathbf{0}\})} \mu_v^t  \psi^v z_v \Big\|_F\\
&\leq  
      \frac{1}{N} \sum_{v\in V\setminus (W\cup \{\mathbf{0}\})}
              \mu^t \| \psi^v \|_2 \, \| z_v\|_2 
=    O\big( \mu^t N \sqrt{d\log dN}\,\big) .
\end{align*}
The lower bound of the lemma
implies that, for any $v\in W$,
\begin{align*}
\Big\|  \sum_{v\in W} \mu_v^t \psi^v z_v \Big\|_F^2
&= \, \mathrm{tr}\, \Big( \sum_{v\in W} \mu_v^t \psi^v z_v \Big)^\mathrm{H}
                  \Big( \sum_{v\in W} \mu_v^t \psi^v z_v \Big) 
  = \, \mathrm{tr}\, \Big\{  \sum_{v\in W} z_v^\mathrm{H} (\psi^v)^\mathrm{H} \psi^v z_v \Big\} \\
&= \,  N \cdot \mathrm{tr}\, \Big\{  \sum_{v\in W} z_v^\mathrm{H} z_v \Big\} 
              = N \sum_{v\in W} \|z_v\|_2^2 \geq \Omega ( 1 ) .
\end{align*}
For large enough $t= \Omega\big(\log (dN)/ \log (1/\mu)\big)$,
the sum in~(\ref{y(t)=zv-eta}) dominates $\eta(t)$ with high probability, while 
the latter decays exponentially fast.
Thus the dynamics $y(t)$ is asymptotically equivalent
to $y^*(t)= \frac{1}{N} \sum_{v\in W} \mu_v^t \psi^v  z_v$.
Recall that $\lambda_v = |\lambda_v|\omega^{\theta_v}$;
since, for $v\in W$, 
$\mu_v= \lambda_v/\lambda$ has modulus~1, it is equal to $\omega^{\theta_v}$.
This implies that
$
y_v^*(t) =
    \frac{1}{N} \sum_{h\in W}
         \,  \sum_{u\in V}
                 \omega^{ t \theta_h + \langle h, v-u \rangle } x_u
$.
Because $y_v^*(t)$ is real, we can ignore the imaginary part 
when expanding the expression above, which completes the proof.
\end{proof}

\subsection{Geometric investigations}

The trajectory $y_v^*(t)$ is called the \emph{limiting orbit}.\footnote{The phase space
of the dynamical system is $\mathbb{R}^{dN}$, but by abuse of notation
we use the word ``orbit'' to refer the trajectory of a single agent, which lies in $\mathbb{R}^d$.}
Theorem~\ref{y(t)=thm} indicates that, with high probability,
every orbit is attracted to its limiting form at an exponential rate,
so we may focus on the latter. Given the initial placement $x$ of the agents,
all the limiting orbits lie in the set $\mathbb{S}$, expressed
in parametric form by
\begin{equation}\label{AtElSum}
\mathbb{S} = \frac{1}{N} \sum_{h\in W}
\big\{ (a_h x) \cos X_h + (b_h x) \sin X_h \big\}.
\end{equation}
Recall that $a_h x $ and $b_h x$ are row vectors in $\mathbb{R}^d$.
The attractor $\mathbb{S}$ is
the Minkowski sum of a number of ellipses.
We examine the geometric structure $\mathbb{S}$
and explain how the limiting orbits embed into it.
To do that, we break up the sum~(\ref{y-abx}) into three parts.
Given $h\in W$, we know that $\lambda_h\neq 0$ by~(\ref{second-eigen}),
so there remain the following cases for the subdominant eigenvalues:

\begin{itemize}
\item
\emph{real} $\lambda_h>0$:
the contribution to the sum is $c_v x$, where $c_v$ is the row vector
\begin{equation}\label{cv}
c_v:= \frac{1}{N} \sum_{h\in W: \, \theta_h=0}
\Big\{\, a_h \cos  \frac{2\pi \langle h, v \rangle}{n} 
        + b_h \sin \frac{2\pi \langle h, v \rangle}{n} \, \Big\}.
\end{equation}
\item
\emph{real} $\lambda_h<0$: 
the contribution is $(-1)^t d_v x$, where, likewise, $d_v$ is the row vector
\begin{equation}\label{dv}
d_v:= \frac{1}{N} \sum_{h\in W: \, \theta_h=n/2}
\Big\{\, a_h \cos  \frac{2\pi \langle h, v \rangle}{n} 
        + b_h \sin \frac{2\pi \langle h, v \rangle}{n} \, \Big\}.
\end{equation}

\item
\emph{nonreal} $\lambda_h$: we can assume that $\theta_h>0$ since the conjugate eigenvalue
$\bar \lambda_h= \lambda_{-h}$ is also present in $W$. The contribution of an eigenvalue
is the same as that of its conjugate since $a_h=a_{-h}$ and $b_h= - b_{-h}$.
So the contribution of a given $\theta>0$ is equal to $e_{v, \theta} x$, where 
\begin{equation*}
e_{v, \theta}:= \frac{2}{N}
\sum_{h\in W: \, \theta_h = \theta}
\left\{\, a_h \cos \frac{2\pi  (t \theta  + \langle h, v \rangle)}{n} 
   +    b_h  \sin  \frac{2\pi (t \theta +  \langle h, v \rangle)}{n} \,\right\} ,
\end{equation*}
which we can expand as
$a_{v,\theta} \cos \frac{2\pi  \theta t }{n}  + b_{v,\theta} \sin \frac{2\pi  \theta t }{n}$,
where\footnote{
$R(\alpha)= \begin{pmatrix}
 \cos \alpha & -\sin \alpha
 \\
\sin \alpha &\cos\alpha
\end{pmatrix}$.} 
\begin{equation}\label{ab_v}
\begin{pmatrix}
a_{v,\theta} \\
b_{v,\theta}
\end{pmatrix}
=  \frac{2}{N}
\sum_{h\in W: \, \theta_h = \theta}
R\left( \frac{-2\pi \langle h, v \rangle}{n}\right)
\begin{pmatrix}
a_h \\
b_h
\end{pmatrix}.
\end{equation}

\end{itemize}
Putting all three contributions together, we find
\begin{equation}\label{limit-orb}
y_v^*(t) = 
c_v x
+ (-1)^t d_v x
+ \sum_{\theta\in \vartheta}
\Big\{\, a_{v,\theta} \cos \frac{2\pi \theta t}{n}  
   +   b_{v,\theta}  \sin  \frac{2\pi \theta t}{n} \, \Big\} x,
\end{equation}  
where $\vartheta$ is the set of distinct $\theta_h>0$ for $h\in W$
and all other quantities are defined in~(\ref{cv}, \ref{dv}, \ref{ab_v}). 
See Figure~\ref{trajectory} for an illustration of a doubly-elliptical orbit around
its torus-like attractor.

\begin{figure}[htb]
\centering
\includegraphics[width=12cm]{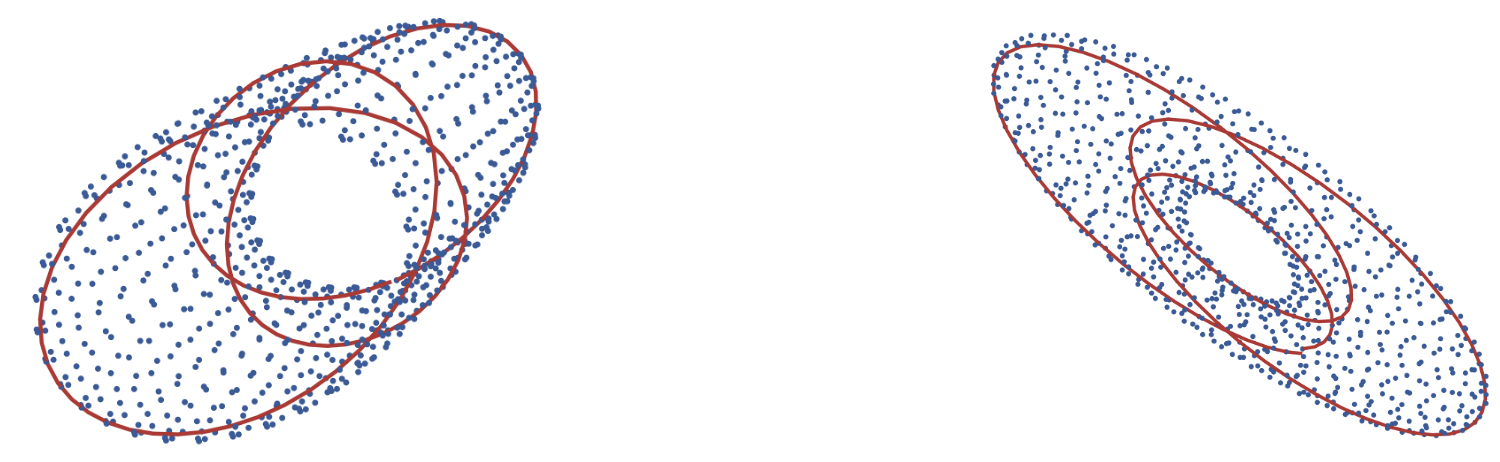}
\caption{Two orbits of a single agent around its attractor.}
\label{trajectory}
\end{figure}

\subsubsection{Generic elliptical attraction}

We prove that, for almost all values of the self-confidence weight $p$,
the set $W$ generates either a single real eigenvalue or two complex conjugate ones.
By~(\ref{limit-orb}), this shows that almost every limiting
orbit is either a single fixed point or a subset of an ellipse in $\mathbb{R}^d$.

\begin{theorem}\label{AS-Ellipse}
There exists a set $\Lambda$ of at most $\binom{N}{2}$ reals
such that the set $W$ 
is associated with either a single real eigenvalue or two complex conjugate ones,
for any $p\in (1/N,1)\setminus \Lambda$.
\end{theorem}

The system is called \emph{regular} if $p\in (1/N,1)\setminus \Lambda$.
For such a system,
either (i) $\vartheta = \{\theta\}$ and $c_v=d_v= \mathbf{0}$,
or (ii) $\vartheta = \emptyset$ and exactly one of $c_v$ or $d_v$ equals $\mathbf{0}$.
In other words, by~(\ref{limit-orb}), we have three cases depending
on the subdominant eigenvalues:
\begin{equation}\label{reg-limit-orb}
y_v^*(t) = 
\begin{cases}
\,\, c_v x & : \text{\small \emph{real positive}}   \\
\,\, (-1)^t d_v x  & : \text{\small \emph{real negative}}   \\
\,\, 
\big(\, a_{v,\theta} \cos \frac{2\pi \theta t}{n}  
   +   b_{v,\theta}  \sin  \frac{2\pi \theta t}{n} \, \big) x
       & : \text{\small \emph{conjugate pair}.}   \\
\end{cases}
\end{equation}  

\begin{lemma}\label{triangleLemma}
Consider a triangle $abc$ and let $e= pc +(1-p)a$
and $f= pc +(1-p)b$. Let $O$ be the origin
and assume that the segments $Oe$ and $Of$
are of the same length (Fig.~\ref{fig-triangleProof}); then the identity
$
|a|^2 - |b|^2 = \frac{2p}{1-p} (b-a)\cdot c
$
holds.
\end{lemma}
\begin{proof}
Let $d:= \frac{1}{2} (e+f)$ be the midpoint of $ef$. The segment $Od$ lies
on the perpendicular bisector of $ef$, so it is orthogonal to
$ef$; hence to $ab$. Thus, $d\cdot (b-a)=0$. 
Since $d=\frac{1}{2}( 2pc +(1-p)a + (1-p)b)$, 
the lemma follows from
$(2pc +(1-p)(a+b))\cdot (b-a)=0$.
\end{proof}

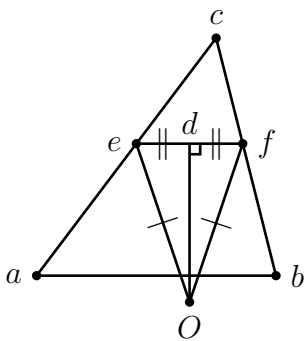
\begin{figure}[htb]
\centering
\usetikzlibrary{positioning}

\begin{tikzpicture}[line width=1pt, scale=0.7]
    \def\circlesize{2.5pt}
    \def\nodefontsize{\Large}

    \fill (0,0) circle (\circlesize) node[below, inner sep=5pt, font=\nodefontsize] {$O$};
    \fill (-1,3) circle (\circlesize) node[left, inner sep=5pt, font=\nodefontsize] {$e$};
    \fill (1,3) circle (\circlesize) node[right, inner sep=5pt, font=\nodefontsize] {$f$};
    \draw (0,0) -- node[sloped, font=\nodefontsize] {$\mid$} (1,3);
    \draw (1,3) -- node[sloped, font=\nodefontsize] {$\parallel$} (0, 3) node[above, font=\nodefontsize] {$d$} -- node[sloped, font=\nodefontsize] {$\parallel$} (-1,3);
    \draw (0,0) -- node[sloped, font=\nodefontsize] {$\mid$} (-1, 3);
    \draw (0, 0) -- (0, 3);
    \draw (0,2.8)-|(0.2,3);

    \fill (0.5,5) circle (\circlesize) node[above, inner sep=5pt, font=\nodefontsize] {$c$};
    \fill (1.625,0.5) circle (\circlesize) node[right, inner sep=5pt, font=\nodefontsize] {$b$};
    \fill (-2.875,0.5) circle (\circlesize) node[left, inner sep=5pt, font=\nodefontsize] {$a$};
    \draw (0.5,5) -- (1.625,0.5) -- (-2.875,0.5) -- (0.5,5);
\end{tikzpicture}
\caption{A triangle identity.}
\label{fig-triangleProof}
\end{figure}

\begin{proof}[Proof of \cref{AS-Ellipse}]
Pick two distinct $u,v\in W$.
Applying Lemma~\ref{triangleLemma} in the complex plane,
we set: 
$a= \frac{1}{|C|} \sum_{h\in C}\omega^{\langle u,h\rangle}$;
$b= \frac{1}{|C|} \sum_{h\in C}\omega^{\langle v,h\rangle}$;
and $c=1$; thus $e= \lambda_u$ and $f= \lambda_v$,
which implies that the segments $Oe$ and $Of$
are of the same length. Abusing notation by treating $a,b,c$ as both vectors
and complex numbers, we have $(b-a)\cdot c= \Re (b-a)$; therefore,
\[
\big(2 \Re (b-a) +|a|^2 - |b|^2\big)p  = |a|^2 - |b|^2.
\]
\begin{enumerate}
\item
If $2 \Re (b-a) +|a|^2 - |b|^2 =0$, then $|a|= |b|$, which in turn
implies that $\Re (b-a) =0$;  hence $a=\bar b$
and $\lambda_u= \bar \lambda_v$.
\item
If $2 \Re (b-a) +|a|^2 - |b|^2 \neq 0$, then
$p$ is unique: $p= p_{u,v}$.
\end{enumerate}

\noindent
We form $\Lambda$ by including all of the numbers $p_{u,v}$, with $u,v\in W$.
\end{proof}

In some cases, regularity holds with no need to exclude values of $p$:

\begin{theorem}\label{Cbasis}
If $C$ forms a basis of $V$ and $n$ is prime,
then $|W|=2m$ and $W$ produces exactly two eigenvalues:
$p + \frac{1-p}{m}(\omega -1)$ and its conjugate.
\end{theorem}
\begin{proof}
By Lemma~\ref{eigenvalues}, 
$\lambda_v= p +  \frac{1-p}{|C|}  \sum_{h\in C} \omega^{\langle v, h\rangle}$.
Fix nonzero $v\in V$. 
Because $n$ is prime and 
the vectors $h_1,\ldots, h_m$ from $C$ form a basis over the field $\mathbb{Z}/n \mathbb{Z}$,
the $m$-by-$m$ matrix whose $i$-th row is $h_i$ is nonsingular.
This implies that, in the sum $\sum_{h\in C} \omega^{\langle v, h\rangle}$, the
exponent sequence $(1,0,\ldots, 0)$ appears for exactly one value $v\in V$
and the same is true of $(-1,0,\ldots, 0)$. 
This holds true 
of any one-hot vector with a single $\pm 1$ at any place and $0$ elsewhere.
This means that, for $2m$ values of $v$,
the eigenvalue $\lambda_v$ is of the form $p +  \frac{1-p}{m}  (m-1+\omega)$
or its conjugate.
Simple examination shows that these are precisely the subdominant eigenvalues.
\end{proof}

\subsubsection{The case of cycle convolutions}

It is useful to consider the case of a single cycle: $m=1$.
For convenience, we momentarily
assume that $n$ is prime and that $\sum_{h\in C} h\neq 0 \pmod{n}$;
both assumptions will be dropped in subsequent sections.

\begin{lemma}\label{Wstructure}
Each eigenvalue $\lambda_v$ is simple.
\end{lemma}
\begin{proof}
Because $n$ is prime, the cyclotomic polynomial for $\omega$
is $\Phi(z)= z^{n-1}+ z^{n-2}+\cdots + z+1$.
It is the minimal polynomial for $\omega$, which is unique.
Note that $\langle v, h\rangle= vh$, since $m=1$.
Given $v\in V$, we define the polynomial 
$g_v(z)=  \sum_{h\in C} z^{vh}$ in 
the quotient ring of rational polynomials $\mathbb{Q}[z]/(z^n-1)$. Sorting the summands
by degree modulo $n$, we have 
$g_v(z)= \sum_{k=0}^{n-1} q_{v,k} z^k$,
for nonnegative integers $q_{v,k}$, where $\sum_k q_{v,k}=|C|$.
If $\lambda_v=\lambda_u$, for some $u\in V$,
then, by Lemma~\ref{eigenvalues}, $g_v(\omega)= g_u(\omega)$; hence $\Phi$ divides $g_v-g_u$.
Because the latter is of degree at most $n-1$, it is either
identically zero or equal to $\Phi$
up to a rational factor $r\neq 0$. In the second case,   
\[
(q_{v,n-1}- q_{u,n-1}) z^{n-1}+\cdots 
      + (q_{v,1}- q_{u,1}) z + q_{v,0}- q_{u,0} = r  \Phi.
\]
This implies that $q_{v,k}- q_{u,k}= r \neq 0$, for all $0\leq k <n$,
which contradicts the fact that $\sum_k q_{v,k} = \sum_k q_{u,k} = |C|$;
therefore, $g_v=g_u$.

\begin{enumerate}
\item
If $v=0$, then $g_v(z)=|C|$; hence $g_u(z)=|C|$ and $u=0$, ie, $v=u$.

\item
If $v\neq 0$, then let $S_v= \{ \omega^{vh}\, |\,  h\in C \}$.
Because $\mathbb{Z}/n\mathbb{Z}$ is a field,
the $|C|$ roots of unity in $S_v$ are distinct;
hence $q_{v,k}\in \{0,1\}$. It follows that $S_v= S_u$ and 
$|S_v|=|S_u|=|C|$;
therefore, for some permutation $\sigma$ of order $|C|$,
we have $vh = u \sigma(h)$, for all $h\in C$.
Summing up both sides over $h\in C$ gives us
$v\sum_{h\in C} h = u\sum_{h\in C} h \pmod{n}$; hence $v=u$,
since $\sum_{h\in C} h\neq 0 \pmod{n}$. \qedhere
\end{enumerate}
\end{proof}
By~(\ref{reg-limit-orb}),
the limiting orbit is of the form $y_v^*(t) =  c_v x$ or
$y_v^*(t)= (-1)^t d_v x$ if the subdominant eigenvalue is real.
Otherwise, the orbit of an agent approaches a single ellipse in $\mathbb{R}^d$:
for some $\theta>0$,
$
y_v^*(t) =
\left(\, a_{v,\theta} \cos \frac{2\pi \theta t}{n}  
   +   b_{v,\theta}  \sin  \frac{2\pi \theta t}{n} \, \right) x
$.

\subsubsection{Opinion velocities}

Assume that the system is regular, so 
$W$ is associated with
either a single real eigenvalue or two complex conjugate ones.
If $\vartheta= \emptyset$, by~(\ref{limit-orb}), 
every agent converges to a fixed point of the attractor $\mathbb{S}$ or its
limiting orbit has a period of 2. 
The other case $\vartheta= \{\theta\}$ is more interesting. The agent
approaches its limiting orbit, which is periodic or quasi-periodic.
The \emph{rotation number}, $\alpha:= \theta/n$,
is the (average) fraction of a whole rotation covered in a single step. 
It measures the speed at which the agent cycles around its orbit.
We prove a lower bound on that speed.\footnote{Its upper bound is $1/2$.}

\begin{theorem}\label{speeds}
The rotation number $\alpha$ of a regular system satisfies
$
\alpha \geq  \frac{1-p}{n} \left(\frac{1}{2N} \right)^N
$.
\end{theorem}
\begin{proof}
Of course, this assumes that $\vartheta\neq \emptyset$.
Fix $v\in V$ and let $\beta_v = \sum_{w\in C} 
      \big(  \omega^{\langle v, w\rangle} - \omega^{-\langle v, w\rangle} \bigr)$;
further, assume that $\beta_v$ is nonzero, hence imaginary.
We have $\beta_v\, \psi_u^v= g_u^\top \psi^v$, where $g_u$ is a vector in $\{-1,0,1\}^{N}$.
It follows that $\beta_v \, \psi^v= A \psi^v$, for an $N$-by-$N$ matrix $A$ whose nonzero elements
are $\pm 1$ and whose rows are given by $g_u^\top$.
Thus, $\beta_v$ is an imaginary eigenvalue of $A$; hence a complex root
of the characteristic polynomial
$\det( A - \gamma \mathbb{I})$. Let $r\geq 1$ be the rank of $A$ and let 
$\gamma_1,\ldots, \gamma_r$ be its nonzero eigenvalues.
Expansion of the determinant gives us a sum of monomials 
of the form $b_i \gamma^{l_i}$, for $1\leq i\leq 2^N N!$,
where $b_i\in \{-1,0,1\}$. The subset of them given by $l_i=N-r$
add up to the product of the nonzero eigenvalues (times $\pm \gamma^{N-r}$);
hence $\prod_{i=1}^r |\gamma_i| \geq 1$.
Label $\gamma_1$ the nonzero eigenvalue of smallest modulus. 
The sum of the squared moduli of the eigenvalues of a matrix is at most
the square of its Frobenius norm; hence
no eigenvalue of $A$ can have a modulus larger than $\sqrt{2N |C|}$ and, therefore
\begin{equation}\label{speedLB}
|\beta_v|\geq 
|\gamma_1|  = \frac{ \prod_{i=1}^r |\gamma_i| }{ \prod_{i=2}^r  |\gamma_i|}
\geq \left(\frac{1}{2 N |C|} \right)^{\frac{r-1}{2}}.
\end{equation}
Since $h\in W$, it follows from~(\ref{second-eigen}) that $0< \lambda= |\lambda_h|<1$.
Thus,
\[
|\theta_h|\geq
|\sin\theta_h|= \frac{|\Im(\lambda_h)|}{|\lambda_h|}
\geq \frac{1-p}{|C|} \, 
           \Big| \Im \Big\{\,  \sum_{w\in C} \omega^{\langle h, w\rangle} \, \Big\} \Big|
\geq \frac{1-p}{2|C|} \, |\beta_h| \, .
\]
With $\lambda_h$ assumed to be nonreal, 
Lemma~\ref{eigenvalues} implies that so is $\beta_h$; hence $\beta_h\neq 0$. 
Applying~(\ref{speedLB}) completes the proof.
\end{proof}

Our next result formalizes the intuitive fact
that self-confidence slows down motion.
Self-assured agents are reluctant to change opinions.

\begin{theorem}\label{p-vs-speed-thm}
The rotation number of a regular system cannot increase with $p$. 
\end{theorem}
\begin{proof}
We must have $|\vartheta|=1$.
Let $\lambda_h$ be (an) eigenvalue corresponding
to the unique angle in~$\vartheta$; recall that $0<\theta_h< n/2$.
As we replace $p$ by $p'>p$, we use the prime sign with all relevant
quantities post-substitution.
Thus, the subdominant eigenvalue for $p'$ associated with
$\vartheta'$ is denoted by $\lambda_v'$;
again, we assume that $|\vartheta'|=1$.
Note that $v$ might not
necessarily be equal to $h$; hence the case analysis:
\begin{itemize}
\item
$v=h$: 
Using the same notation for complex numbers and the points
in the plane they represent (Fig.~\ref{p-vs-speed}),
we see that $\lambda_h'$ lies in (the relative interior of) the segment $1\lambda_h$;
hence $\theta_h'<\theta_h$.
\item
$v\neq h$:
We prove that, as illustrated in Fig.~\ref{p-vs-speed},
all three conditions
$|\lambda_h|>|\lambda_v|$,
$|\lambda_h'|<|\lambda_v'|$, and  
$\theta_h< \theta_v' \leq n/2$,
cannot hold at the same time, which will establish
the theorem.
If we increase $q$ continuously from $p$ to $p'$, 
$\theta_h(q)$ decreases continuously.  (We use the argument $q$ to denote
the fact that $\theta_h$ corresponds to the eigenvalue defined with $p$ replaced by $q$.)
Since, at the end of that motion,
$|\lambda_h(q)|<|\lambda_v(q)|$, by continuity we have
$p_o < p'$, where $p_o= \min\{ q>p\,:\, |\lambda_h(q)| = |\lambda_v(q)|\}$.
To simplify the notation, we repurpose the use of the prime superscript
to refer to $p_o$ (eg, $p'=p_o$). So, we now have 
$|\lambda_h'| = |\lambda_v'|$ and $\theta_h< \theta_v'< \theta_v\leq n/2$.
It follows that (i) the point $\lambda_v$ lies in the pie slice of radius $|\lambda_h|$ 
running counterclockwise from $\lambda_h$ to $-|\lambda_h|$ on the real axis.
Also, because $|\lambda_h'| = |\lambda_v'|$ and $|\lambda_h|>|\lambda_v|$,
setting $c=1$ as before in Lemma~\ref{triangleLemma} shows that
(ii) $\Re(\lambda_v) > \Re(\lambda_h)$.\footnote{The keen-eyed observer will
notice that in the lemma we must plug in $(p_o-p)/(1-p)$ instead of $p$.}
Putting (i, ii) together shows that $\theta_h\geq n/4$ (as shown in Fig.~\ref{p-vs-speed}).
Consequently, 
the slope of the segment $\lambda_h \lambda_v$ is negative. Since
that segment is parallel to $\lambda_h' \lambda_v'$, the perpendicular
bisector of the latter has positive slope. Since that bisector
is above $\lambda_v'$ and $\Im(\lambda_v') \geq 0$, this implies that $0$
and $\lambda_h'$ are on opposite sides of that bisector; hence
$|\lambda_v'| < |\lambda_h'|$, which is a contradiction. \qedhere
\end{itemize}
\end{proof}

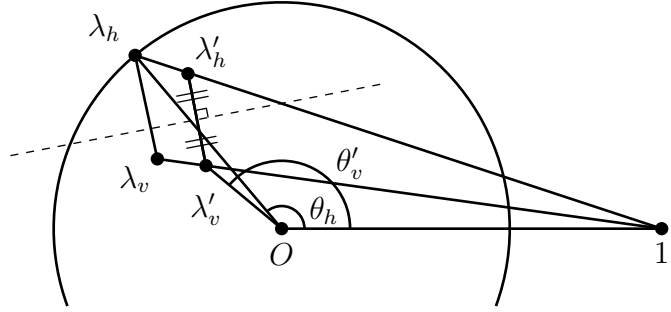
\begin{figure}[htb]
\centering
\usetikzlibrary{positioning}
\usetikzlibrary{calc}

\begin{tikzpicture}[line width=1pt]
    \def\circlesize{2.5pt}
    \def\nodefontsize{\Large}

    \draw (0, 0) + (-20:3) arc (-20:200:3);
    \fill (0,0) circle (\circlesize) node[below, inner sep=5pt, font=\nodefontsize] {$O$};
    \fill (5,0) circle (\circlesize) node[below, inner sep=5pt, font=\nodefontsize] {$1$};
    \draw (0, 0) -- (5, 0);

    \fill (0,0)+(130:3) circle (\circlesize) node[above left, inner sep=5pt, font=\nodefontsize] {$\lambda_h$};
    \draw (0,0)+(130:3) -- (0,0);
    \draw (0,0)+(130:3) -- (5,0);
    \draw (0, 0) + (0:0.3) node[above right, font=\nodefontsize, xshift=-2pt, yshift=-2pt] {$\theta_h$} arc (0:130:0.3);

    \fill (0,0)+(140:1.3) circle (\circlesize) node[below=5pt, inner sep=5pt, font=\nodefontsize] {$\lambda_v'$};
    \draw (0,0)+(140:1.3) -- (0,0);
    \draw (0, 0) + (0:0.9) node[above=15pt, font=\nodefontsize] {$\theta_v'$} arc (0:140:0.9);

    \draw (5,0)+(172.066:6.7) -- (5,0);
    \fill (5,0)+(172.066:6.7) circle (\circlesize) node[below left, inner sep=2pt, font=\nodefontsize] {$\lambda_v$};
    \draw ([shift={(172.066:6.7)}] 5,0) -- ([shift={(130:3)}] 0,0);
    \draw ([shift={(140:1.3)}] 0,0) -- ++(101:1.242);
    \fill ([shift={(140:1.3)}] 0,0) +(101:1.242) circle (\circlesize) node[above right, inner sep=2pt, font=\nodefontsize] {$\lambda_h'$};

    \coordinate (start) at ([shift={(140:1.3)}] 0,0);
    \coordinate (end) at ($(start) + (101:1.242)$);
    \coordinate (midpoint) at ($(start)!0.5!(end)$);
    \coordinate (perpendicular1) at ($(midpoint)!-2.5cm!90:(start)$);
    \coordinate (perpendicular2) at ($(midpoint)!2.5cm!90:(start)$);
    \draw[line width=0.5pt, dashed] (perpendicular1) -- (perpendicular2);

    \draw (start) -- node[sloped, font=\nodefontsize] {$\parallel$} (midpoint) -- node[sloped, font=\nodefontsize] {$\parallel$} (end);
    
    \coordinate (upperP) at ($(midpoint)!0.2!(end)$);
    \coordinate (rightP) at ($(midpoint)!0.06!(perpendicular2)$);
    \coordinate (upperrightP) at ($(rightP) + (upperP) - (midpoint)$);
    \draw[line width=0.5pt] (upperP)--(upperrightP)--(rightP);
\end{tikzpicture}
\caption{Why self-confidence slows down the dynamics: proof
by contradiction.}
\label{p-vs-speed}
\end{figure}

\subsection{Equidistributed orbits}\label{uniformOrbits}

The attractor $\mathbb{S}$ is the Minkowski sum of a number of ellipses
bounded by $|W|$. 
An agent orbits around an ellipse as it gets attracted to it
exponentially fast.
In a regular system with $\vartheta\neq \emptyset$,
its limiting orbit is periodic if the unique angle $\theta_h$ of $\vartheta$
is rational; it is quasi-periodic otherwise. 
In fact, it then forms a dense subset of the ellipse.  By~(\ref{limit-orb}), this follows
from Weyl's ergodicity principle~\cite{kuipersN}, which states that  
the set $\{\alpha t \pmod{1} ,|\, t\geq 0\}$ is uniformly distributed in $[0,1)$,
for any irrational~$\alpha$.

Dropping the regularity requirement
may produce more exotic dynamics. We exhibit instances
where a limiting orbit will not only be dense over the entire attracting set  
but, in fact, uniformly distributed. 
In other words, an agent will approach \emph{every possible opinion} with equal frequency.
This will occur when this property holds:\footnote{The coordinates of $a=(a_1,\ldots, a_k)$ 
are linearly independent over the rationals if $\mathbf{0}$ is the only rational vector 
normal to $a$.}

\begin{assumption} \label{lin-indep}
The numbers in $\vartheta\cup\{1\}$ are linearly
independent over the rationals. 
\end{assumption}

We explain this phenomenon next. 
Order the angles of $\vartheta$ arbitrarily and define the 
vector $\alpha = (\alpha_1,\ldots, \alpha_s)\in \big[0, \frac{1}{2}\big]^s$, where
$s= |\vartheta|$ and $\alpha_j= \theta/n$ for the $j$-th angle $\theta\in \vartheta$. 
We may assume that $c_v=d_v=\mathbf{0}$ in~(\ref{limit-orb}) since
these cases are rotationally trivial.  By Assumption~\ref{lin-indep}, 
$\mathbf{0}$ is the only integer vector 
whose dot product with $\alpha$ is an integer.
We use the standard notation $\|\alpha\|_{\,\mathbb{R}/\mathbb{Z}}
= \max_{k\leq s} \min_{a\in \mathbb{Z}}|\alpha_k -a|$.
By Kronecker's approximation theorem~\cite{cassels57},
for any $\beta\in [0,1]^{s}$ and any $\eps>0$,
there exists $q\in \mathbb{Z}$ such that $\|q\alpha - \beta \|_{\,\mathbb{R}/\mathbb{Z}}\leq \eps$.
It follows directly that, with high probability, any limiting orbit is dense 
over the attractor~$\, \mathbb{S}$.
We prove the stronger result:

\begin{theorem}\label{U-equid}
Under Assumption~\ref{lin-indep}, any limiting orbit is uniformly distributed over 
the attractor~$\, \mathbb{S}$.
\end{theorem}

We mention that, in general, Assumption~\ref{lin-indep} might be difficult to verify analytically.
Empirically, however, density
is fairly obvious to ascertain numerically with suitable visual evidence~(Fig.~\ref{fig-equid}).
\begin{figure}[htb]
\centering
\includegraphics[width=9cm]{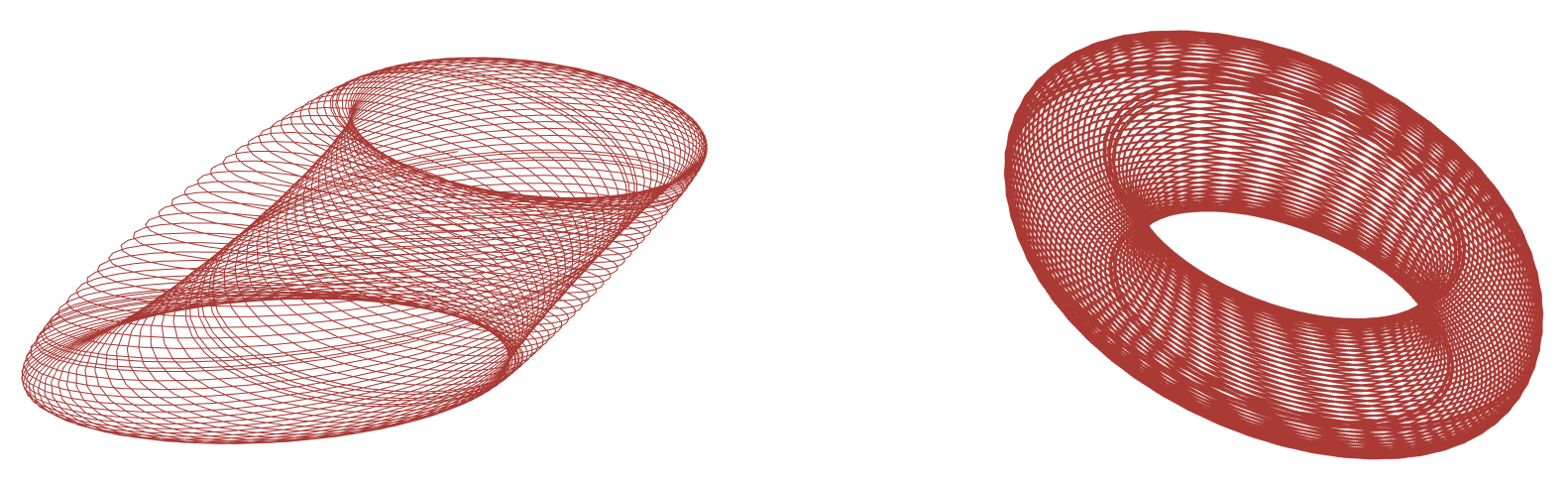}
\caption{Two examples where an agent approaches every point on its attractor
with equal frequency. In each case, the 
curve traces the orbit of the agent.}
\label{fig-equid}
\end{figure}
We define the discrepancy $D(S_t)$ of 
$S_t = (p_1, \ldots, p_t)$, 
with $p_i \in \mathbb{R}^s$, as
\[
D(S_t)= \sup_{B \in J} \Big|\, \frac{A(B;t)}{t} - \mu_s(B)\, \Big|,
\]
where $\mu_s$ is the $s$-dimensional Lebesgue measure
and $A(B;t)= |\{i\,|\, p_i\in B\}|$
and $J$ is the set of $s$-dimensional boxes of the form
$\prod_{i=1}^s \, [a_i, b_i)\subset [0,1]^s$.
The infinite sequence $S_\infty$ is said to be \emph{uniformly distributed}
if $D(S_t)$ tends to $0$, as $t$ goes to infinity.

\begin{lemma}\label{ETK}
{\sc{(Erd{\H o}s--Tur\'an--Koksma}~\cite{kuipersN},}
\text{\emph{page 116}}{\sc ).} \ \ 
For any integer $L>0$, 
\[
D(S_t)  \leq 2s^2 3^{s+1}
       \Big(\, \frac{1}{L} + \sum_{0 < \|\ell\|_{\infty} \leq L} \frac{1}{r(\ell)}\,
                     \Big| \frac{1}{t} \sum_{k=1}^{t} e^{ 2\pi i \langle \ell, p_k \rangle } \Big|\, \Big) \, ,
\]
where
$ r(\ell): = \prod_{j=1}^{s} \max\{1, |\ell_j|\}$
and $\ell = (\ell_1, \ldots, \ell_s) \in \mathbb{Z}^s$.
\end{lemma}

\begin{proof}[Proof of \cref{U-equid}]
We form the sequence $p_1,\ldots, p_t \in [0,1)^s$ such that
$p_k= k \alpha \pmod{1}$;
where each coordinate of $k\alpha$ is replaced by its fractional part.
By Lemma~\ref{ETK}, its box discrepancy satisfies
\begin{equation}\label{Dt-ineq}
D(S_t) \leq 2s^2 3^{s+1}
       \Big(\, \frac{1}{L} + \sum_{0 < \|\ell\|_{\infty} \leq L} \frac{1}{r(\ell)}\,
                     \Big| \frac{1}{t} \sum_{k=1}^{t} e^{ 2\pi i \langle \ell, k \alpha \rangle } \Big|\, \Big),
\end{equation}
By Assumption~\ref{lin-indep}, $\mathbf{0}$ is the only integer vector 
whose dot product with $\alpha$ is an integer; hence
$\gamma_\ell:= e^{ 2\pi i \langle \ell,  \alpha \rangle }\neq 1$, for any $\ell\neq \mathbf{0}$.
It follows that
\[
\Big| \sum_{k=1}^{t} e^{ 2\pi i \langle \ell, k \alpha \rangle }\, \Big|
= \Big| \sum_{k=1}^{t} \gamma_\ell^k \,\Big|
= \Big| \frac{\gamma_\ell-\gamma_\ell^{t+1}}{1-\gamma_\ell} \Big|
\leq  \frac{2}{|1-\gamma_\ell|} \, .
\]
By~(\ref{Dt-ineq}), for any $\delta>0$,
\[
D(S_t) \leq 2s^2 3^{s+1}
       \Big(\, \frac{1}{L} + \frac{1}{t}
            \sum_{0 < \|\ell\|_{\infty} \leq L} \,
                     \frac{2}{|1-\gamma_\ell| } \, \Big) \leq \delta.
\]
for $L = \lceil 4s^2 3^{s+1}/\delta \rceil$ and
$t\geq (8/\delta) s^2 3^{s+1}  
          \sum_{0 < \|\ell\|_{\infty} \leq L}  |1-\gamma_\ell|^{-1}$.  
\end{proof}

\subsection{Examples}

We illustrate the range of contrarian opinion dynamics by considering
a few specific examples for which calculations are feasible.

\subsubsection{Fixed-point attractor}
 
Set $m=2$ and $C= \{(1,0), (0,1), (-1,0), (0,-1)\}$.
By Lemma~\ref{eigenvalues}, for any $v= (v_1,v_2)\in V$,
\[
\lambda_v= p +  \frac{1-p}{2} 
\Big(\cos \frac{2\pi v_1}{n} + \cos \frac{2\pi v_2}{n} \Big). 
\]
The eigenvalues are real and
$\lambda= \max_{v\in V}\{ |\lambda_v|<1\}=   p +  \frac{1}{2}(1-p)(1+ \cos 2\pi /n)$.
For any $h\in C$, $\lambda_h = \lambda$ and $\theta_h=0$; hence $C\subseteq W$.
A simple examination shows that, in fact, $W=C$.
By~(\ref{cv}, \ref{limit-orb}), given $j\in [d]$,\footnote{As usual, $[d]$ denotes $\{1,\ldots, d\}$.} 
\[
y_v^*(t)_j 
=  A_j\cos \frac{2\pi (v_1 + \alpha_j)}{n}  + B_j\cos \frac{2\pi (v_2 + \beta_j)}{n}\, ,
\]
where $A_j, B_j, \alpha_j, \beta_j$ do not depend on $v$ but only on the initial position $x$.
This produces a 2D surface in $\mathbb{R}^d$ formed by 
the Minkowski sum of two ellipses centered at the origin~(Fig.\ref{fig-ex1}).

\begin{figure}[htb]
\centering
\includegraphics[width=4cm]{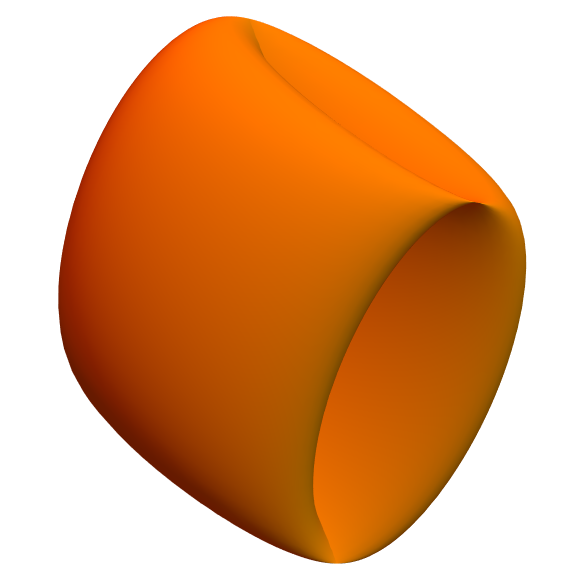}
\caption{The attractor on which each agent converges to a fixed point.}
\label{fig-ex1}
\end{figure}

\subsubsection{Periodic and quasi-periodic orbits}

Set $m=2$ and $C= \{(1,0), (0,1)\}$. 
By Lemma~\ref{eigenvalues}, for any $v\in V$, $
\lambda_v= p +  \frac{1-p}{2} 
\big( \omega^{v_1} + \omega^{v_2} \big)$;
hence
$\lambda = \max_{v\in V}\{ |\lambda_v|<1\}
    =  \frac{1}{2} \big| 1+p + (1-p)\omega \big|$
and $W = \{(1, 0), (0, 1), (-1, 0), (0, -1)\}$.
Specifically, $\lambda_v$ is equal to $\frac{1}{2} \big( 1+p + (1-p)\omega \big)$,
for $v\in \{(1,0), (0,1)\}$, and to its conjugate, for $v\in \{(-1,0), (0,-1)\}$.  
By~(\ref{ab_v}, \ref{limit-orb}), we have $\vartheta = \{\theta\}$, where
\[
\theta 
= \Big(\frac{n}{2\pi}\Big) \arctan \left( \frac{ (1-p)\sin 2\pi/n }{ 1+p + (1-p)\cos 2\pi/n } \right) \, ,
\]
and
\[
y_v^*(t) = 
\Big(\, a_{v,\theta} \cos \frac{2\pi \theta t}{n}  
   +   b_{v,\theta}  \sin  \frac{2\pi \theta t}{n} \, \Big) x \, .
\] 
Fix a coordinate $j\in [d]$; we find that
\[
y_v^*(t)_j 
=  A_j\cos \frac{2\pi (\theta t + v_1 + \alpha_j)}{n}  + B_j\cos \frac{2\pi (\theta t + v_2 + \beta_j)}{n}\, ,
\]
for suitable reals $A_j, B_j, \alpha_j, \beta_j$ that 
depend on the initial position $x$ but not on $v$.
This again produces a two-dimensional attracting subset of $\mathbb{R}^d$ 
formed by the Minkowski sum of two ellipses.
In the case of Figure~\ref{fig-ex2}, the attractor is a torus pinched along
two curves.
The main difference from the previous case comes from the 
limit behavior of the agents.
They are not attracted to a fixed point but, rather, to a surface.
With high probability, the orbits are asymptotically periodic if $\theta$ is rational,
and quasi-periodic otherwise. 
For a case of the former, consider $p=0$, which gives us
\[
\theta
=   \Big(\frac{n}{2\pi}\Big)  \arctan \left( \frac{ \sin 2\pi/n }{ 1+ \cos 2\pi/n } \right) = 
\frac{1}{2}\, ;
\]
hence periodic orbits.

\begin{figure}[htb]
\centering
\includegraphics[width=12cm]{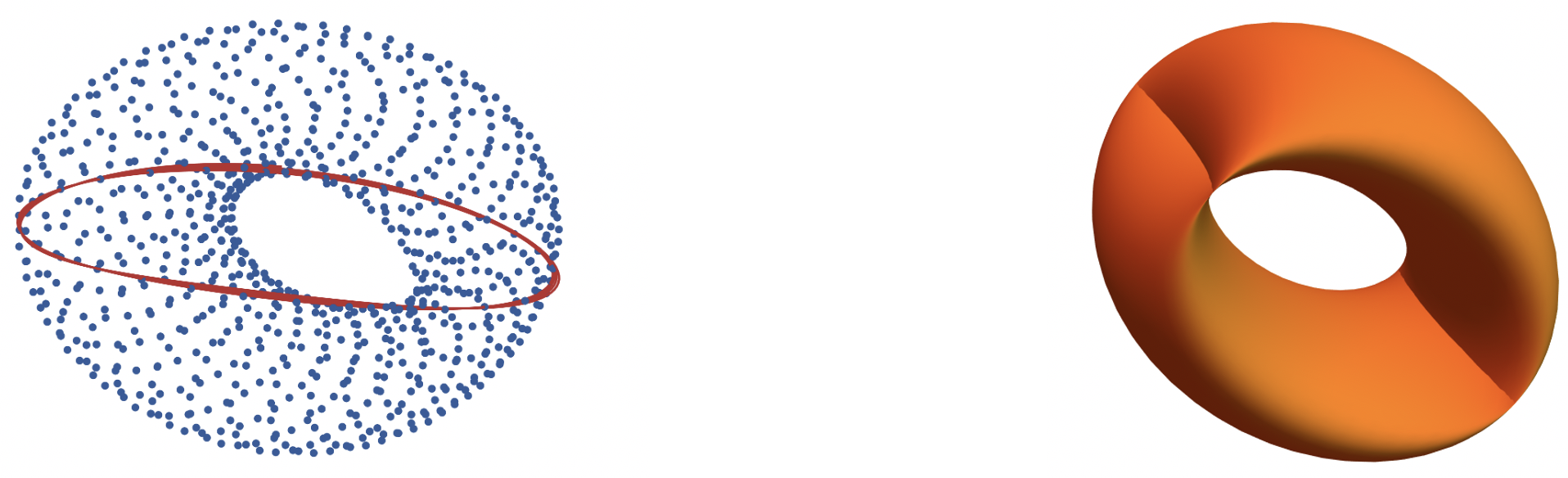}
\caption{A periodic orbit on the left with the full attractor on the right.}
\label{fig-ex2}
\end{figure}

\subsubsection{Equidistribution over the attractor}

Put $m=2$ and $C=\{(1,0), (0,1), (2,3)\}$. We set $p=1/4$.
For any $v\in V$, we have
\[
\lambda_v = p + \frac{1-p}{3}\big( \omega^{v_1} + \omega^{v_2} + \omega^{2v_1+3v_2}\big).
\]
We verified numerically that
$W= \{ (1,0),  (1,-1),  (-1,0),  (-1,1)\}$ and $\vartheta= \{\theta_1, \theta_2\}$, where
\begin{equation*}
\begin{cases}
\,\, \theta_1
=   \big(\frac{n}{2\pi}\big)  
     \arctan \left( \frac{ \sin 2\pi/n  + \sin 4\pi/n}{ 2 + \cos 2\pi/n + \cos 4\pi/n} \right) \\
\,\, \theta_2
=   \big(\frac{n}{2\pi}\big)  
     \arctan \left( \frac{ -\sin 2\pi/n }{ 1 + 3\cos 2\pi/n} \right) .
\end{cases}
\end{equation*}
By~(\ref{limit-orb}),
\[
y_v^*(t) = 
\sum_{k=1,2} \Big(\, a_{v,\theta_k} \cos \frac{2\pi \theta_k t}{n}  
   +   b_{v,\theta_k}  \sin  \frac{2\pi \theta_k t}{n} \, \Big) x \, .
\] 
Computer experimentation points to the linear independence
of the numbers $1, \theta_1, \theta_2$ over the rationals.
If so, then Assumption~\ref{lin-indep} from Section~\ref{uniformOrbits} holds
and, by Theorem~\ref{U-equid},
any limiting orbit is uniformly distributed over 
the attractor~$\, \mathbb{S}$~(Fig.\ref{fig-ex3}).

\begin{figure}[htb]
\centering
\includegraphics[width=4.5cm]{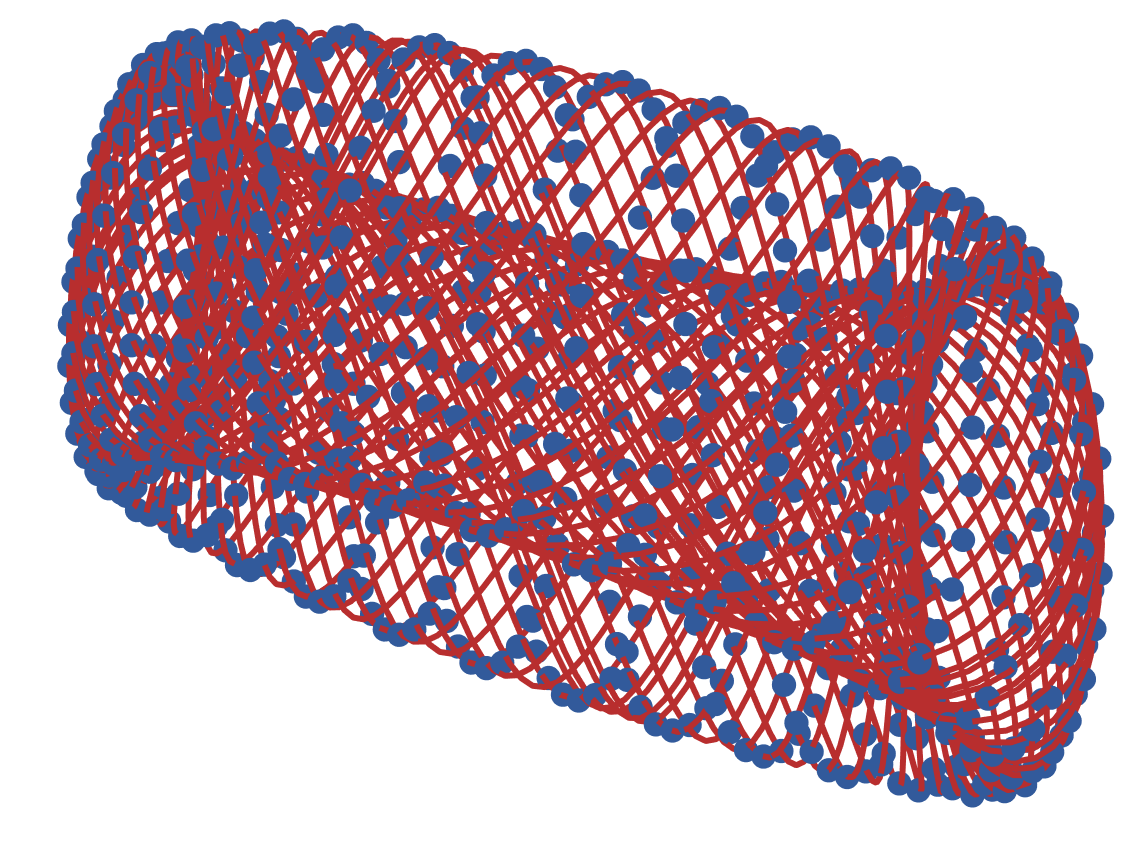}
\caption{A single agent's orbit is uniformly distributed around its attractor.}
\label{fig-ex3}
\end{figure}

\section{Dynamic Social Networks}

We define a \emph{mixed} model of contrarian opinion dynamics.
Let $\mathcal{M}=\{C_1, \ldots, C_s\}$ be a set of $s$ nonempty subsets,
each one spanning the vector space $V$. 
At each time step $t$, we define the matrix $F_C$ by
choosing, as convolution set $C$, a random,
uniformly distributed element of $\mathcal{M}$.
As before, we assume that $\mathbf{1}^\top x=0$.
Let $\lambda_{j,v}$ be the eigenvalue of $F_{C_j}$ associated with $v\in V$.
Given an infinite sequence $I_\infty$ of indices from $[s]$,
we denote
by $I_t={k_1,\ldots, k_t}$ be the first $t$ indices of $I_\infty$, 
and we write 
$\Lambda_v(I_t) =
    \prod_{k\in I_t }  \lambda_{k,v}$.
We generalize~(\ref{x(t)=}) into

\begin{equation}\label{x(t)=rand}
x(t)= \frac{1}{N} \sum_{v\in V\setminus \{\mathbf{0}\}} 
      \Lambda_v(I_t) \, \psi^v z_v \, ,
\end{equation}
where $z_v$ is the row vector $\sum_{u\in V} \omega^{-\langle v,u\rangle} x_u$.

\subsection{Spectral decomposition}

Write $\lambda_v^\times = \big|\prod_{j=1}^{s}  \lambda_{j,v}\big|^{1/s}$
and $\lambda= \max_{ v\in V\setminus \{\mathbf{0}\}} \lambda_v^\times$;
because all the eigenvalues other than $\lambda_{j,\mathbf{0}}=1$
lie strictly inside the unit circle, we have
$\lambda <1$.  Without loss of generality, we can assume that $\lambda>0$.
Indeed, suppose that $\lambda=0$; then, for every $v\in V\setminus \{\mathbf{0}\}$,
there is $j=j(v)$ such 
that $\lambda_{j,v}=0$.  This presents us with a ``coupon collector's'' scenario:
with probability at most $N(1-1/s)^t\leq Ne^{-t/s}$,
we have $\Lambda_v(I_t)\neq 0$ for at least one nonzero $v\in V$.
In other words, with high probability, every coordinate of $x(t)$ in the eigenbasis will vanish
after $O(s\log N)$ steps; hence $x(t)=0$ for all $t$ large enough.
This case is of little interest, so we dismiss it
and assume that $\lambda$ is positive.
We redefine $W=\{v\in V  \,|\, \lambda_v^\times =\lambda \}$.
Let $W'= \{v\in V  \,|\, \lambda_v^\times <\lambda \}$.

\begin{lemma}\label{ratio}
If $W'$ is nonempty, there exists $c<1$ such that, with high probability,
for all $t$ large enough,
\[
\max_{w'\in W'} |\Lambda_{w'}(I_t)|  \leq c^t  \min_{w\in W}|\Lambda_w(I_t)| .
\]
\end{lemma}
Note that the high-probability event applies to \emph{all} times $t$ larger than a fixed constant.
The proof involves the comparison of two multiplicative random walks.
\begin{proof}
Fix $w\in W$ and $w'\in W'$. We prove that
$|\Lambda_{w'}(I_t)|  \leq c^t |\Lambda_w(I_t)|$.
If $\lambda_{w'}^\times =0$, then $\lambda_{j,w'}=0$, for some $j$.
With high probability, the sequence $I_t$ includes the index $j$ at least once
for any $t$ large enough; hence $|\Lambda_{w'}(I_t)| =0$ and the lemma holds.
Assume now that $\lambda_{w'}^\times > 0$;  for all $j$,
both of $\lambda_{j,w}$ and $\lambda_{j,w'}$ are nonzero.
Write $S_v(I_t)= \log \prod_{k\in I_t}|\lambda_{k,v}|$, for $v=w,w'$,
and note that 
$S_v(I_t)= t\log \lambda_v^\times + \sum_{k\in I_t}\sigma_{k,v}$,
where $\sigma_{k,v}= \log|\lambda_{k,v}| - \log \lambda_v^\times$.
Let $\sigma= \max_{k,v} |\sigma_{k,v}|$.
The random variables $\sigma_{k,v}$ are unbiased and mutually independent
in $[-\sigma, \sigma]$. Classic deviation bounds~\cite{AlonSpencer}
give us
$
\mathbb{P}\Big [\, \Big|\sum_{k\in I_t} \sigma_{k,v} \Big|>b \,\Big ]< 2e^{-b^2 / (2t \sigma^2)}.
$
It follows that
$\big| S_v(I_t) - t\log \lambda_v^\times \big|= O\big( \sigma \sqrt{t \ln (tN)}\, \big)$ 
with probability $1-a/(tN)^2$, for an arbitrarily small constant $a>0$.
Since $\sum_{t>0}1/t^2= \pi^2/6$, it follows that, for arbitrarily small fixed $\eps>0$
and all $t$ large enough,
with probability at least $1- \eps/N^2$, 
\[
\log\frac{ |\Lambda_{w}(I_t)| }{ |\Lambda_{w'}(I_t)| }
= S_w(I_t) - S_{w'}(I_t)
\geq  t \log\frac{ \lambda_w^\times }{ \lambda_{w'}^\times } 
     - O\big( \sigma \sqrt{t \log (tN)}\,\big)  
\geq \frac{t}{2} \log\frac{ \lambda_w^\times }{ \lambda_{w'}^\times } \, ,
\]
for any given $w\in W$ and $w'\in W'$.
Setting $c= \max_{w\in W, w'\in W'} 
\sqrt{ \lambda_{w'}^\times / \lambda_w^\times}$ 
and using a union bound
completes the proof.
\end{proof}

We define the scaled orbit $y(t)= x(t)/\lambda^t$. 
Reprising the argument from Theorem~\ref{y(t)=thm},
we conclude from~(\ref{x(t)=rand}) that, with high probability,
the limiting orbit is of the form
\begin{equation*}
y^*(t)= \frac{1}{N} \sum_{h\in W} 
      \Big(\prod_{k\in I_t }  \frac{\lambda_{k,h}}{\lambda}  \Big) \psi^h z_h
      =
         \frac{1}{N} \sum_{h\in W} 
               \Big( \prod_{k\in I_t }    \frac{|\lambda_{k,h}|}{\lambda} \, \Big) \,
                       \omega^{\sum_{k\in I_t} \! \theta_{k,h}}\, \psi^h z_h ,
\end{equation*}
where $\lambda_{k,h} :=  |\lambda_{k,h}| \omega^{\theta_{k,h}}$.
It follows that
\begin{equation*}
y_v^*(t) =
    \frac{1}{N} \sum_{h\in W}
     \Big( \prod_{k\in I_t }    \frac{|\lambda_{k,h}|}{\lambda} \, \Big) \,
      \omega^{ \sum_{k\in I_t} \! \theta_{k,h} + \langle h, v \rangle }
         \,  \sum_{u\in V}
                 \omega^{- \langle h, u \rangle } x_u \, .
\end{equation*}
If we put 
$X_h= \frac{2\pi}{n} \big( \langle h, v \rangle + \sum_{k\in I_t} \! \theta_{k,h}\big)$,
then, with $a_h$ and $b_h$ being the row vectors defined in Theorem~\ref{y(t)=thm},
\begin{equation}\label{y-rand-abx}
y_v^*(t) = \frac{1}{N}
\sum_{h\in W}
 \Big( \prod_{k\in I_t }    \frac{|\lambda_{k,h}|}{\lambda} \, \Big) \,
\Big( (a_hx)\cos X_h  +  (b_h x) \sin  X_h\Big).
\end{equation}

\subsection{Surprising attractors}

Adding mixing to a model increases the entropy of the system.
It is thus to be expected that the attractor of a mixed model
should have higher dimensionality than its pure components.
The surprise is that this need not be the case. We exhibit
instances of contrarian opinion dynamics where mixing
\emph{decreases} the dimension of the attractor.
To keep the notation simple, we consider two pure models 
$\mathcal{M}_1=\{C_1\}$,  $\mathcal{M}_2=\{C_2\}$
alongside their mixture $\mathcal{M}_3=\{C_1, C_2\}$.

\begin{theorem}\label{surprise1}
For any $k\in [m]$, there is a choice of $C_1$ and $C_2$ such that
$\mathrm{dim}\, \mathcal{M}_3 = k$ and
$\mathrm{dim}\, \mathcal{M}_1 = \mathrm{dim}\, \mathcal{M}_2 = m$;
in other words, the dimension of the mixture's attractor can be
arbitrarily smaller than those of its pure components.
\end{theorem}
\begin{proof}
We define $C_1=  (e_1,\ldots, e_m)$ 
and $C_2= (e_1,\ldots, e_k, 2e_{k+1},\ldots, 2e_m)$, 
for any $k\in [m]$, 
where $e_i$ is the one-hot vector of $V$ whose $i$-th coordinate
is 1 and all the others 0.
Let $W_i$ denote the set $W$ corresponding to the system $\mathcal{M}_i$.
We easily verify that $W_1= \pm C_1$ and 
$W_2= \pm \big\{ e_1,\ldots, e_k,  2^{-1} e_{k+1},\ldots, 2^{-1}e_m \big\}$,
where $2^{-1}$ is the inverse of $2$ in the field $\mathbb{Z}/n \mathbb{Z}$.
A vector $v\in W_i$ and its negative contribute to the same ellipse, so
we have $\mathrm{dim}\, \mathcal{M}_1 = \mathrm{dim}\, \mathcal{M}_2 =m$.
We note that $|\lambda_{k,h}|= \lambda= \big|1 +\frac{1-p}{m}(\omega -1)\big|$,
for $h\in  W_1 \cup W_2$; hence
$\lambda_v^\times = \lambda$ for $h\in  W_1 \cap W_2$
and $\lambda_v^\times < \lambda$ for all other values of $h$.
It follows that  $\mathrm{dim}\, \mathcal{M}_3 = k$.
\end{proof}

Figure~\ref{low-mixattract} illustrates Theorem~\ref{surprise1}.
We have $m=2$ and $n= 29$.
The two convolution sets 
are $C_1=\{(1,0), (0,1)\}$ and $C_2=\{(1,0), (0,2)\}$.
The initial positions are random and identical
in all three cases.

\begin{figure}[htb]
\centering
\includegraphics[width=11cm, height=3cm]{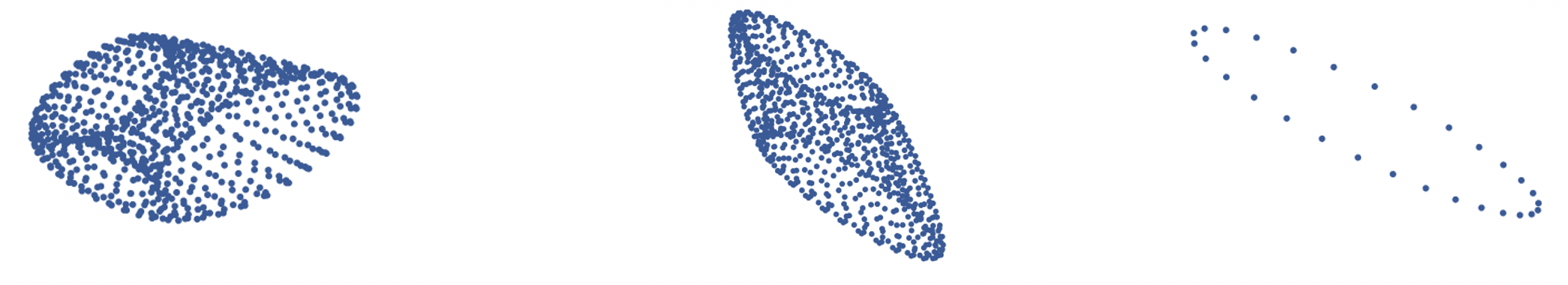}
\caption{The two attractors of the pure models
$\mathcal{M}_i$ ($i=1,2$) on the left,
with the lower-dimensional attractor of the mixture on the right.}
\label{low-mixattract}
\end{figure}

We can generalize the mixed model by picking $C_1$ (resp. $C_2$) 
with probability $1-q$ (resp. $q$), where $0\leq q\leq 1$.  For this, we redefine 
$\lambda_v^\times(q) = \big| \lambda_{1,v}^{1-q}\lambda_{2,v}^q \big|$
and $\lambda(q)= \max_{ v\in V\setminus \{\mathbf{0}\}} \lambda_v^\times(q)$.

\begin{theorem}\label{surprise2}
There is a choice of $C_1$ and $C_2$ such that
$\mathrm{dim}\, \mathcal{M}_3 > \mathrm{dim}\, \mathcal{M}_1 = \mathrm{dim}\, \mathcal{M}_2 = m$;
in other words, the dimension of the mixture's attractor can be
larger than those of its pure components.
\end{theorem}
\begin{proof}
Borrowing the notation of the previous proof,
we define $C_1=  (e_1,\ldots, e_m)$ and $C_2=  (2e_1,\ldots, 2e_m)$
and verify that $W_1= \pm C_1$ and $W_2= \pm \{2^{-1} e_1,\ldots, 2^{-1}e_m \big\}$;
hence $\mathrm{dim}\, \mathcal{M}_1 = \mathrm{dim}\, \mathcal{M}_2 = m$.
Assuming that $n>3$, we note that the sets $W_1$ and $W_2$ are disjoint.
Regarding the mixed system, we have
$W(q)= \{v\in V\,:\, \lambda_v^\times(q) =  \lambda(q)\}$, where $W(0)= W_1$ and $W(1)= W_2$.
Around $q=0$, we have, for all $v\in W(0)$,
\begin{equation}\label{W(0)}
\lambda_v^\times(q) =
    \Big| \,1 +\frac{1-p}{m}(\omega -1)\, \Big|^{1-q}
        \times   \Big|\, 1 +\frac{1-p}{m}(\omega^2 -1) \, \Big|^q \, .
\end{equation}
Since $W(0)\neq W(1)$, by continuity, there are $q\in (0,1)$ 
and $w\in  W(q)\setminus W(0)$ such that
$\lambda_w^\times(q)$ is equal to the right-hand side of~(\ref{W(0)}).
This implies that $W(q)\supseteq W(0)\cup \{w\}$,
which completes the proof.
\end{proof}

Figure~\ref{high-mixattract} illustrates the theorem.
We have $m=2$, $n= 29$, and $p=0.9$.
The two convolution sets 
are $C_1=\{(1,0), (0,1)\}$ and $C_2=\{(2,0), (0,2)\}$;
the mixture probability is $q= 0.0306$.  
The initial positions are random and identical
in all three cases.

\begin{figure}[htb]
\centering
\includegraphics[width=10cm]{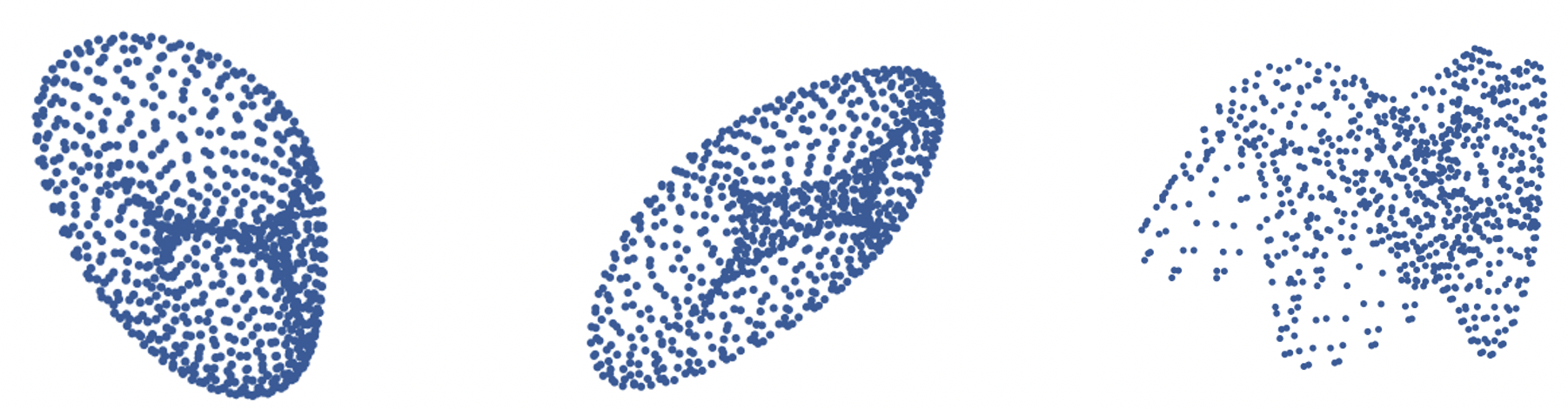}
\caption{The two attractors of the pure models
$\mathcal{M}_i$ ($i=1,2$) on the left,
with the higher-dimensional attractor of the mixture on the right.}
\label{high-mixattract}
\end{figure}


\newpage
\bibliography{ref}

\end{document}